\documentclass[conference,letterpaper]{IEEEtran}

\usepackage[top=0.52in, bottom=0.52in, left=0.6in, right=0.6in]{geometry}
\normalsize
\usepackage{cite}
\usepackage{amsmath,amssymb,amsfonts,amsthm}
\usepackage{ifthen}
\usepackage{tikz}
\usepackage{cite}
\usepackage{url}
\usepackage{titlesec}
\usepackage{subfig}
\usepackage{blkarray}
\usepackage{dsfont}
\usepackage[mathscr]{euscript}
\usepackage{enumitem}
\usepackage{algorithm}
\usepackage[noend]{algpseudocode}
\usepackage{blkarray}
\usepackage{stfloats}%
\newtheorem{theorem}{Theorem}[section]
\newtheorem{corollary}{Corollary}[section]

\theoremstyle{remark}
\newtheorem*{remark}{Remark}
\theoremstyle{definition}
\newtheorem{definition}{Definition}[section]

\usepackage{url}
\usepackage{arydshln}
\usepackage{mathtools}
\usepackage{dsfont}

\definecolor{brickred}{cmyk}{0,0.89,0.94,0.28}
\definecolor{goldenrod}{cmyk}{0,0.10,0.84,0}
\definecolor{purple}{cmyk}{0.45,0.86,0,0}
\definecolor{rawsienna}{cmyk}{0,0.72,1,0.45}
\definecolor{olivegreen}{cmyk}{0.64,0,0.95,0.40}
\definecolor{peach}{cmyk}{0,0.5,0.7,0}
\definecolor{darkolive}{rgb}{0.,0.4,0.}
\colorlet{grey}{gray!40}

\def\0{{\mathbf 0}}
\def\1{{\mathbf 1}}

\usepackage[cmintegrals]{newtxmath}
\newlist{mylist}{enumerate}{3}
\setlist[mylist]{label={}}
\usepackage{graphicx}
\usepackage{textcomp}
\def\BibTeX{{\rm B\kern-.05em{\sc i\kern-.025em b}\kern-.08em
		T\kern-.1667em\lower.7ex\hbox{E}\kern-.125emX}}

\ifCLASSINFOpdf
\else
\fi

\begin{document}
%
\title{Estimating the Weight Enumerators of Reed-Muller Codes via Sampling}
%
%
%

\author{
	\IEEEauthorblockN{Shreyas Jain}
	\IEEEauthorblockA{Dept. of Mathematical Sciences\\ IISER Mohali\\
		Email: ms20098@iisermohali.ac.in}
	\and
	\IEEEauthorblockN{V.~Arvind~Rameshwar}
	\IEEEauthorblockA{India Urban Data Exchange\\
		Email: arvind.rameshwar@gmail.com}
	\and
	\IEEEauthorblockN{Navin Kashyap}
	\IEEEauthorblockA{Dept. of ECE\\ IISc, Bengaluru\\
		Email: nkashyap@iisc.ac.in}
	\thanks{This work was carried out when V.~A.~Rameshwar was a Ph.D. student at the Dept. of ECE, IISc, Bengaluru. The work of V.~A.~Rameshwar was supported in part by a Prime Minister’s Research Fellowship from the Ministry of Education, Govt. of India and in part by a Qualcomm Innovation Fellowship. The work of N.~Kashyap was supported in part by the grant DST/INT/RUS/RSF/P-41/2021 from the Department of Science \& Technology, Govt. of India.}}
\IEEEoverridecommandlockouts
%
%

\markboth{}%
{Rameshwar, Jain, and Kashyap: Sampling-Based Weight Enumerator Estimates of RM Codes}
%



\maketitle

\begin{abstract}
This paper develops an algorithmic approach for obtaining estimates of the weight enumerators of Reed-Muller (RM) codes. Our algorithm is based on a technique for estimating the partition functions of spin systems, which in turn employs a sampler that produces codewords according to a suitably defined Gibbs distribution. We apply our method to  moderate-blocklength RM codes and derive approximate values of their weight enumerators. We observe that the rates of the weight enumerator estimates returned by our method are close to the true rates when these rates are either known or computable by brute-force search; in other cases, our computations provide provably robust estimates. As a byproduct, our sampling algorithm also allows us to obtain estimates of the weight spectra of RM codes. We illustrate our methods by providing estimates of the hitherto unknown weight enumerators of the RM$(11,5)$ code and the exact weight spectra of the RM$(10,3)$ and RM$(10,4)$ codes.
\end{abstract}


%
\IEEEpeerreviewmaketitle

\section{Introduction}
%
%
%
%
\IEEEPARstart{R}{eed}-Muller (RM) codes are a family of binary linear codes that are obtained by the evaluations of Boolean polynomials on the points of the Boolean hypercube. These algebraic codes have been of interest to practitioners for several decades, for their applications in deep-space to 5G cellular communications (see, e.g., \cite{deepspace,5g1}). Furthermore, recent breakthrough theoretical progress \cite{Reeves,abbesandon} has shown that RM codes are in fact capacity-achieving for general binary-input memoryless symmetric (BMS) channels, under both bitwise maximum a-posteriori probability (bit-MAP) and blockwise maximum a-posteriori probability (block-MAP) decoding. 


Despite extensive study on the RM family of codes, a basic property that is still not completely understood is their weight enumerators, or the number of codewords having a given weight $\omega \in \{0,\ldots,N\}$, where $N$ is the blocklength of the code under consideration. Early progress on the weight distribution of RM codes, or the collection of weight enumerators, was made in \cite{kasamitokura,kasami2.5d}, which characterized all codewords of weight up to $2.5d$, where $d$ is the minimum weight of the code under consideration. Other works \cite{sloaneberlekamp,sugino,sugita} computed numerical values of, or analytical expressions for, the weight distribution of specific RM codes. Much later, a series of works \cite{kaufman,sam,sberlo,asw} derived analytical bounds on the weight enumerators of RM codes using ideas from the analysis of Boolean functions on the hypercube. More recently, the work \cite{anuprao} proposed simple upper bounds on the weight enumerators using the symmetry properties of the RM family of codes. On the algorithmic front, the techniques used for computing the weight enumerators of fixed blocklength RM codes primarily draw from \cite{sarwate}, which provides a recursive algorithm based on computing the weight enumerators of cosets of small RM codes that lie inside larger RM codes. However, this recursive procedure quickly becomes computationally intractable for even moderate-blocklength (blocklength around $1000$ symbols) RM codes.


In this paper, we seek to obtain approximate, numerical estimates of the weight enumerators, via sampling techniques. In particular, via estimates of the sizes of constant-weight subcodes of RM codes, we arrive at estimates of the weight distribution of moderate-blocklength RM codes, thereby making progress on a wide-open research problem. Our technique makes use of a simple statistical physics approach for estimating the partition functions of spin systems. Such an approximate counting technique has been largely unexplored in the coding theory literature, and we believe that there is scope for its broader application to other problems of interest in coding theory (see \cite{gamal} for some early work on using simulated annealing as a heuristic for identifying good short-blocklength error-correcting codes). The crux of this approach is the employment of a Monte-Carlo Markov Chain (MCMC) sampler that draws codewords according to a suitably biased Gibbs distribution. Importantly, our sampler involves a ``nearest-neighbour'' proposal distribution, which uses minimum-weight codewords of RM codes. Our sampler, for sufficiently large ``inverse temperature'' parameters, can produce samples from exponentially-small (compared to the size of the parent RM code) subcodes of RM codes, and is hence of independent interest. We mention that we have recently employed such techniques for estimating the sizes of runlength limited constrained subcodes of RM codes in \cite{arsjnk24ncc}. In this paper, we also present a simple algorithm that uses our sampler for estimating the weight spectra of RM codes, or the collection of weights with positive weight enumerators.

As illustrations of our methods, we provide estimates of (rates of) the weight enumerators of moderate-blocklength (blocklength $N\leq 2048$) RM codes, and compare these estimates with the true rates that are either known from the literature or computed via exact counting algorithms. We then use our algorithm to obtain estimates of the hitherto-unknown weight enumerators of the RM$(11,5)$ code and also derive the exact weight spectrum (also unknown) of the RM$(10,3)$ and RM$(10,4)$ codes. In the appendix, we provide theoretical guarantees of the sample complexity of our weight enumerator estimation algorithm (for a fixed error in approximation), and demonstrate that the number of samples, and hence the time taken to run the estimation algorithm, is only polynomial in the blocklength of the RM code, and is independent of its dimension.

\section{Preliminaries and Notation}
We use $\mathbb{F}_2$ to denote the binary field, i.e., the set $\{0,1\}$ equipped with modulo-$2$ arithmetic. 
We use bold letters such as $\mathbf{x}$, $\mathbf{y}$ to denote finite-length binary sequences (or vectors); the set of all finite-length binary sequences is denoted by $\{0,1\}^\star$. 
Further, when $\mathbf{x}, \mathbf{y} \in \mathbb{F}_2^n$, we denote by $\mathbf{x}+\mathbf{y}$ the vector resulting from component-wise modulo-$2$ addition. We also use the notation $w_H(\mathbf{x})$ to denote the Hamming weight of $\mathbf{x}\in \mathbb{F}^n$, which is the number of nonzero coordinates in $\mathbf{x}$. We define the indicator function of a set $\mathcal{A}$ as $\mathds{1}_\mathcal{A}$, with $\mathds{1}_\mathcal{A}(\mathbf{x}) = 1$, when $\mathbf{x}\in \mathcal{A}$, and $0$, otherwise. We use exp$(x)$ to denote $e^x$, for $x\in \mathbb{R}$. For sequences $(a_n)_{n\geq 1}$ and $(b_n)_{n\geq 1}$ of positive reals, we say that $a_n = O(b_n)$, if there exists $n_0\in \mathbb{N}$ and a positive real $M$, such that $a_n\leq M\cdot b_n$, for all $n\geq n_0$. We say that $a_n = \Theta(b_n)$, if there exist positive reals $M_1, M_2$ such that $M_1\cdot b_n \leq a_n\leq M_2\cdot b_n$ for all sufficiently large $n$. 

\subsection{Reed-Muller Codes}
\label{sec:rmintro}
We now recall the definition of the binary Reed-Muller (RM) family of codes and some of their basic facts that are relevant to this work. Codewords of binary RM codes consist of the evaluation vectors of multivariate polynomials over the binary field $\mathbb{F}_2$. Consider the polynomial ring $\mathbb{F}_2[x_1,x_2,\ldots,x_m]$ in $m$ variables. Note that any polynomial $f\in \mathbb{F}_2[x_1,x_2,\ldots,x_m]$ can be expressed as the sum of {monomials} of the form $\prod_{j\in S:S\subseteq [m]} x_j$, since $x^2 = x$ over the field $\mathbb{F}_2$. For a polynomial $f\in \mathbb{F}_2[x_1,x_2,\ldots,x_m]$ and a binary vector $\mathbf{z} = (z_1,\ldots,z_m)\in \mathbb{F}_2^m$, we write {$f(\mathbf{z})=f(z_1,\ldots,z_m)$} as the evaluation of $f$ at $\mathbf{z}$. The evaluation points are ordered according to the standard lexicographic order on strings in $\mathbb{F}_2^m$, i.e., if $\mathbf{z} = (z_1,\ldots,z_m)$ and $\mathbf{z}^{\prime} = (z_1^{\prime},\ldots,z_m^{\prime})$ are two evaluation points, then, $\mathbf{z}$ occurs before $\mathbf{z}^{\prime}$ iff for some $i\geq 1$, we have $z_j = z_j^{\prime}$ for all $j<i$, and $z_i < z_i^{\prime}$. Now, let Eval$(f):=\left({f(\mathbf{z})}:\mathbf{z}\in \mathbb{F}_2^m\right)$ be the evaluation vector of $f$, where the coordinates $\mathbf{z}$ are ordered according to the standard lexicographic order. 

\begin{definition}[see Chap. 13 in \cite{mws}, or \cite{rm_survey}]
	{For $0\leq r\leq m$}, the $r^{\text{th}}$-order binary Reed-Muller code RM$(m,r)$ is defined as
	\[
	\text{RM}(m,r):=\{\text{Eval}(f): f\in \mathbb{F}_2[x_1,x_2,\ldots,x_m],\ \text{deg}(f)\leq r\},
	\]
	where $\text{deg}(f)$ is the degree of the largest monomial in $f$, and the degree of a monomial $\prod_{j\in S: S\subseteq [m]} x_j$ is simply $|S|$. 
\end{definition}

It is known that the evaluation vectors of all the distinct monomials in the variables $x_1,\ldots, x_m$ are linearly independent over $\mathbb{F}_2$. Hence, RM$(m,r)$ has dimension $\binom{m}{\le r} := \sum_{i=0}^{r}{m \choose i}$. We then have that RM$(m,r)$ is a $\left[2^m,\binom{m}{\le r}\right]$ linear code. {Furthermore, the dual code of RM$(m,r)$ is RM$(m,m-r-1)$, for $m\geq 1$ and $r\leq m-1$.} It is also known that RM$(m,r)$ has minimum Hamming distance ${d}_{\text{min}}(\text{RM}(m,r))=2^{m-r}$. Of importance in this paper is the fact that each minimum-weight codeword of RM$(m,r)$ is the characteristic vector of an $(m-r)$-dimensional affine subspace of $\mathbb{F}_2^m$ (see \cite[Thm. 8, Chap. 13, p. 380]{mws}) , i.e., the vector having ones in those coordinates $\mathbf{z}\in \{0,1\}^m$ that lie in the affine subspace. We use this fact to efficiently sample  a uniformly random minimum-weight codeword of a given RM code. Another fact of use to us is that the collection of minimum-weight codewords spans RM$(m,r)$, for any $m\geq 1$ and $r\leq m$ (see \cite[Thm. 12, Chap. 13, p. 385]{mws}).
\section{Sampling-Based Algorithms}
\label{sec:estimate}
In this section, we discuss our sampling-based procedure for computing estimates of weight enumerators of RM codes, or equivalently, of the sizes of constant-weight subcodes of RM codes. As a by product, we also obtain a simple algorithm for computing estimates of the weight spectrum. We first present a general approach for obtaining the weight enumerator estimates for a given $[n,k]$ linear code $\mathcal{C}$ of blocklength $n$ and dimension $k$. We mention that while the technique described here applies to \emph{any} code $\mathcal{C}$ (not necessarily linear), we restrict our attention to RM codes in this paper, since in this family of codes, it is possible to generate samples from the distribution $p_\beta$ (see \eqref{eq:pb}) in a computationally efficient manner, as will be required later (see Section \ref{sec:sampler}). 

Given an $[n,k]$ linear code $\mathcal{C}$, for any integer $0\leq \omega\leq n$, let $\mathcal{C}^{(\omega)}\subseteq \mathcal{C}$ denote the set of codewords of Hamming weight exactly $\omega$. We are interested in obtaining estimates of the quantity
$
	A(\omega):= |\mathcal{C}^{(\omega)}| = \sum_{\mathbf{c}\in \mathcal{C}} \mathds{1}\{\mathbf{c}\in W^{(\omega)}\},
$
where $W^{(\omega)}\subseteq \{0,1\}^n$ is the collection of all length-$n$ binary sequences of weight $\omega$. Further, let $S$ denote the weight spectrum of $\mathcal{C}$, i.e., $S = \{0\leq \omega\leq n: A(\omega)>0\}$. For a fixed $\omega$, we call the quantity $A(\omega)$ as $Z$, to establish similarity with the notation for the partition function of the probability distribution $p$, where
\[
p(\mathbf{x}) = \frac{1}{Z}\cdot \mathds{1}_{\mathcal{C}^{(\omega)}}(\mathbf{x}),\ \mathbf{x}\in \{0,1\}^n.
\]

When the dimension $k$ of $\mathcal{C}$ is large, a direct calculation of $Z$ is computationally intractable, as it involves roughly $\min\{2^k,2^{n-k}\}$ additions \cite{macwilliams}. \footnote{We mention that the algorithmic question of deciding if $A(\omega)>0$ (resp. computing $A(\omega)$) is equivalent to deciding if a suitably defined constraint satisfaction problem (CSP) has a solution (resp. counting the number of solutions to the CSP). There is vast literature on the computational complexity of CSPs (see, e.g., \cite{chen1,chen2} and references therein). One interesting direction of research will be to comment on the complexity of computing $A(\omega)$ for RM$(m,r)$, when say $r$ grows with $m$, using the dichotomy results in \cite{schaefer,creignou}.} We first focus on computing good approximations to $Z$; as a byproduct, we will come up with a simple algorithm for estimating $S$.
%


Before we do so, we need some additional background and notation. Consider the following probability distribution supported on the codewords of $\mathcal{C}$:
\begin{equation}
	\label{eq:pb}
p_\beta(\mathbf{x}) = \frac{1}{Z_\beta}\cdot e^{-\beta\cdot E(\mathbf{x})}\cdot \mathds{1}_\mathcal{C}(\mathbf{x}), \quad \mathbf{x}\in \{0,1\}^n,
\end{equation}
where $\beta>0$ is some fixed real number (in statistical physics, $\beta$ is termed as ``inverse temperature''), and $E: \{0,1\}^n\to [0,\infty)$ is an ``energy function'' such that $E(\mathbf{x}) = 0$ if $\mathbf{x}\in {W}^{(\omega)}$ and is strictly positive, otherwise. For a given weight $0\leq \omega\leq n$, we define the energy function $E = E^{(\omega)}$ as
$
E^{(\omega)}(\mathbf{x}) = |w_H(\mathbf{x})-\omega|.
$
Note that the ``partition function'' or the normalization constant
\begin{equation}
	\label{eq:Zb}
	Z_\beta = \sum_{\mathbf{c}\in \mathcal{C}} e^{-\beta\cdot E(\mathbf{x})}.
\end{equation}
In the limit as $\beta\to \infty$, it can be argued that the distribution $p_\beta$ becomes the uniform distribution over the ``ground states'' or zero-energy vectors in $\mathcal{C}$ (see, e.g., \cite[Chapter 2]{mezard}); more precisely,
\begin{equation}
	\label{eq:limit}
\lim_{\beta\to \infty} p_\beta(\mathbf{x}) = p(\mathbf{x}).
\end{equation}
Clearly, by the definition of the energy function, we also have that $Z = \lim_{\beta\to \infty}Z_\beta$. We shall use this perspective to compute an approximation to $Z$. 


\subsection{Algorithms for Computing $Z_{\beta^\star}$}
\label{sec:practice}
In this section, following \eqref{eq:limit}, we shall use the partition function $Z_{\beta^\star}$, where $\beta^\star$ is suitably large, as an estimate of $Z$ (the question of how large $\beta^\star$ must be for $Z_{\beta^\star}$ to be a good approximation to $Z = A(\omega)$ is taken up in Appendix \ref{sec:theory}). To accomplish this, we shall outline a fairly standard method from the statistical physics literature \cite{statphy} (see also Lecture 4 in \cite{sinclair}) to compute $Z_{\beta^\star}$ corresponding to the two constraints defined above, when $\beta^\star$ is large.   We provide a qualitative description of the procedure in this section; the exact values of the parameters required to guarantee a close-enough estimate will be provided in Section I of the appendix.

The key idea in this method is to express $Z_{\beta^\star}$ as a telescoping product of ratios of partition functions, for smaller values of $\beta$. We define a sequence (or a ``cooling schedule'')
$
0 = \beta_0< \beta_1< \ldots < \beta_\ell= \beta^\star,
$
where $\beta_i = \beta_{i-1}+\frac{1}{n}$, $1\leq i\leq \ell$, and $\ell$ is a large positive integer, and write 
\begin{equation}
	\label{eq:zbstar}
	Z_{\beta^\star} = Z_{\beta_0}\times\prod_{i=1}^{\ell}\frac{Z_{\beta_i}}{Z_{\beta_{i-1}}}.
\end{equation}
Observe from \eqref{eq:Zb} that $Z_{\beta_0} = Z_0 = |\mathcal{C}|=2^k$, where $k$ is the dimension of $\mathcal{C}$. Now, for $1\leq i\leq\ell$, we have from \eqref{eq:pb} that
\begin{align}
	\frac{Z_{\beta_i}}{Z_{\beta_{i-1}}} &=  \frac{1}{Z_{\beta_{i-1}}}\sum_{\mathbf{c}\in \mathcal{C}} \text{exp}(-\beta_iE(\mathbf{c})) \notag\\
	&= \frac{1}{Z_{\beta_{i-1}}}\sum_{\mathbf{c}\in \mathcal{C}}\text{exp}(-\beta_{i-1}E(\mathbf{c}))\cdot \text{exp}((\beta_{i-1}-\beta_i)E(\mathbf{c})) \notag\\
	&= \mathbb{E}[\text{exp}(-E(\mathbf{c})/n)] \label{eq:expectation},
\end{align}
where the expectation is over codewords $\mathbf{c}$ drawn according to $p_{\beta_{i-1}}$. In other words, the ratio $\frac{Z_{\beta_i}}{Z_{\beta_{i-1}}}$ can be computed as the expected value of a random variable $X_i:=\text{exp}(-E(\mathbf{c})/n)$, where $\mathbf{c}$ is drawn according to $p_{\beta_{i-1}}$. A description of how such a random variable can be sampled is provided in Algorithm \ref{alg:sample} in Section \ref{sec:sampler}; it is here that we use the fact that $\mathcal{C}$ is an RM code, which allows for efficient sampling from $p_\beta$, for a given $\beta$ (Algorithm \ref{alg:sample} also gives us a simple algorithm for estimating the weight spectrum). We now explain how \eqref{eq:expectation} is used to obtain an estimate of $Z_{\beta^\star}$ in \eqref{eq:zbstar}. For every $i$, for large $t$, we sample i.i.d. random variables $X_{i,j}$, $1\leq j\leq t$, which have the same distribution as $X_i$. We ensure that the $X_{i,j}$s are independent across $i$ as well. We then estimate the expected value in \eqref{eq:expectation} by a sample average, i.e., we define the random variable
\begin{equation}
	\label{eq:sampleav}
Y_i:= \frac{1}{t}\sum_{j=1}^t X_{i,j}.
\end{equation}
Finally, the estimate for $Z_{\beta^\star}$ (see \eqref{eq:zbstar}) that we shall use is
\begin{equation}
	\label{eq:estimate}
\widehat{Z}_{\beta^\star} = Z_{\beta_0}\times \prod_{i=1}^\ell Y_i.
\end{equation}
Note that we then have, by independence of the $X_{i,j}$s and hence of the $Y_i$s, that $\mathbb{E}[\widehat{Z}_{\beta^\star}] = Z_0\times \prod_{i=1}^\ell \mathbb{E}[Y_i] = Z_{\beta^\star}$. A summary of our algorithm is shown as Algorithm \ref{alg:estimate}. 

In Appendix \ref{sec:theory}, we argue that it suffices to set $\beta^\star = \Theta(n^2)$ and $t = \Theta(n^3)$ to guarantee that $\widehat{Z}_{\beta^\star}$ is close to $Z = A(\omega)$.


\begin{algorithm}[t]
	\caption{Estimating $Z$ via $Z_{\beta^\star}$}
	\label{alg:estimate}
	\begin{algorithmic}[1]	
		\Procedure{Estimator}{$\beta^\star$}
		\State Fix a cooling schedule $0=\beta_0<\beta_1<\ldots<\beta_\ell = \beta^\star$.
		\State Fix a large $t\in \mathbb{N}$.
		\For{$i=1:\ell$}
		\State Use Algorithm \ref{alg:sample} to generate $t$ i.i.d. samples $\mathbf{c}_{i,1},\ldots,\mathbf{c}_{i,t}$.
		\State For $1\leq j\leq t$, set $X_{i,j} \leftarrow \text{exp}((\beta_{i-1}-\beta_i)E(\mathbf{c}_{i,j}))$.
		\State Compute $Y_i = \frac{1}{t}\sum_{j=1}^t X_{i,j}$.
		\EndFor
		\State Output $\widehat{Z}_{\beta^\star} = |\mathcal{C}|\times \prod_{i=1}^\ell Y_i$.
		\EndProcedure	
	\end{algorithmic}
\end{algorithm} 
\subsection{An Algorithm for Sampling RM Codewords According to $p_\beta$}
\label{sec:sampler}
Our approach to generating samples from the distribution $p_\beta$, when $\mathcal{C}$ is a Reed-Muller code is a simple ``nearest-neighbour'' Metropolis algorithm, which is a special instance of Monte Carlo Markov Chain (MCMC) methods (see Chapter 3 in \cite{mcmcbook}).

Let $\Delta$ be the collection of minimum-weight codewords in $\mathcal{C}$. Consider the following ``symmetric proposal distribution'' $\{P(\mathbf{c}_1,\mathbf{c}_2):\ \mathbf{c}_1,\mathbf{c}_2\in \mathcal{C}\}$, where $P(\mathbf{c}_1,\mathbf{c}_2)$ is the conditional probability of ``proposing'' codeword $\mathbf{c}_2$ given that we are at codeword $\mathbf{c}_1$:
\begin{equation}
	\label{eq:proposal}
P(\mathbf{c}_1, \mathbf{c}_2) = 
\begin{cases}
	\frac{1}{|\Delta|},\ \text{if $\mathbf{c}_2 = \mathbf{c}_1+\overline{\mathbf{c}}$, for some $\overline{\mathbf{c}}\in \Delta$},\\
	0,\ \text{otherwise}.
\end{cases}
\end{equation}
Clearly, $P$ is symmetric in that $P(\mathbf{c}_1, \mathbf{c}_2) = P(\mathbf{c}_2, \mathbf{c}_1)$, for all $\mathbf{c}_1, \mathbf{c}_2\in \mathcal{C}$. Our Metropolis algorithm begins at a randomly initialized codeword. When the algorithm is at codeword $\mathbf{c}_1$, it ``accepts'' the proposal of codeword $\mathbf{c}_2$ with probability $\min\left(1,\frac{p_\beta(\mathbf{c}_2)}{p_\beta(\mathbf{c}_1)}\right)$, and moves to $\mathbf{c}_2$. Now, observe that since $\mathcal{C}$ is a Reed-Muller code, it is easy to sample a codeword $\mathbf{c}_2$ that differs from $\mathbf{c}_1$ by a minimum-weight codeword; in other words one can efficiently sample a uniformly random minimum-weight codeword $\overline{\mathbf{c}}$ (see Section \ref{sec:rmintro} for a characterization of minimum-weight codewords of RM codes). This sampling procedure is shown as Steps 5--6 in Algorithm \ref{alg:sample}, with $\overline{\mathbf{c}}$ in Step 7 representing the minimum-weight codeword sampled. 
Note that the full-rank matrix $A$ in Step 5 of the algorithm can be sampled by using a rejection sampling procedure (see, for example, Appendix B.5 in \cite{mcmcbook}), in a number of steps that is a constant, in expectation (see sequence A048651 in \cite{oeis}).
 
It can be checked that $p_\beta$ is indeed a stationary distribution of this chain. Further, suppose that $\mathbf{c}^{(\tau)}$ is the (random) codeword that this chain is at, at time $\tau\in \mathbb{N}$. Then, it is well-known that if the Metropolis chain is irreducible and aperiodic (and hence ergodic), then the distribution of $\mathbf{c}_\tau$ is close, in total variational distance, to the stationary distribution $p_\beta$ (see, e.g., Theorem 4.9 in \cite{mcmcbook}), for large enough $\tau$.\footnote{In this work, we do not address the question of how large $\tau$ must be, but simply set $\tau$ to be large enough so that the Metropolis chain reaches the ``zero-energy'' constrained codewords within $\tau$ steps, in practice, starting from an arbitrary initial codeword.}

Now, since the set of minimum-weight codewords $\Delta$ spans $\mathcal{C}$, we have that the chain is irreducible. Further, for some selected weights $\omega$, we can argue that there always exists a pair of codewords $(\mathbf{c}_1,\mathbf{c}_2)$ such that  $\mathbf{c}_2 = \mathbf{c}_1+\overline{\mathbf{c}}$, for some $\overline{\mathbf{c}}\in \Delta$, with $p_\beta(\mathbf{c}_2)<p_\beta(\mathbf{c}_1)$. We then get that $Q(\mathbf{c}_1,\mathbf{c}_1)>0$, assuring us of aperiodicity, and hence of ergodicity, of our chain. However, for the purposes of this work, we do not concern ourselves with explicitly proving aperiodicity, and instead seek to test the soundness of our technique, numerically. 
 \begin{algorithm}[t]
 	\caption{Sampling RM codewords approximately from $p_\beta$}
 	\label{alg:sample}
 	\begin{algorithmic}[1]	
 		\Procedure{Metropolis-Sampler}{$\mathbf{c}^{(0)}$, $\beta$, $E$}       
 		\State Initialize the Metropolis chain at the arbitrary (fixed) codeword $\mathbf{c}^{(0)}$.
 		\State Fix a large $\tau\in \mathbb{N}$.
 		\For{$i=1:\tau$}
 		\State Generate a uniformly random $(m-r)\times m$ full-rank $0$-$1$ matrix $A$ and a uniformly random vector $\mathbf{b}\in \{0,1\}^n$.
 		\State Construct $H = \{\mathbf{z}: \mathbf{z}=\mathbf{x}\cdot A+\mathbf{b},\text{ for some $\mathbf{x}\in \mathbb{F}_2^{m-r}$}\}$.
 		\State Set $\overline{\mathbf{c}}$ to be the characteristic vector of $H$ and set $\mathbf{c}\leftarrow \mathbf{c}^{(i-1)}+\overline{\mathbf{c}}$.
 		\State Set $\mathbf{c}^{(i)} \leftarrow \mathbf{c}$ with probability $\min\left(1,\text{exp}(-\beta(E(\mathbf{c})-E(\mathbf{c}^{(i-1)})))\right)$; else set $\mathbf{c}^{(i)} \leftarrow \mathbf{c}^{(i-1)}$.
 		\EndFor
 		\State Output $\mathbf{c}_\tau$.
 		\EndProcedure	
 	\end{algorithmic}
 \end{algorithm} 

\subsection{Algorithm for Computing Weight Spectrum Estimate}
Besides allowing us to estimate the weight enumerator $Z = A(\omega)$, for a fixed $0\leq\omega\leq n$ via $\widehat{Z}_{\beta^\star}$, the sampling procedure in Algorithm \ref{alg:sample} allows us check if there exists a codeword of weight $\omega$ or not; in other words, the sampler allows us to obtain the weight spectrum of the code $\mathcal{C}$. This procedure relies on the fact that for a fixed, large value of $\beta^\star$, the distribution $p_{\beta^\star}$ is close (in total variational distance) to $p$ (see \eqref{eq:limit}), which is supported only on codewords of weight $\omega$. We hence simply draw a codeword $\mathbf{c}$ from $p_{\beta^\star}$ and check if $w_H(\mathbf{c}) = \omega$; from the previous observation, we expect that with high probability this is indeed true. Our algorithm for checking if there exists a codeword of weight $\omega$ is given as Algorithm \ref{alg:check}. Our estimate for the true weight spectrum $S$ is the set $\hat{S}$ that aggregates all weights $\omega$ where Algorithm \ref{alg:check} outputs \texttt{Yes}. Clearly, $\hat{S}\subseteq S$.

\begin{algorithm}[t]
	\caption{Checking if $\omega$ has a positive weight enumerator}
	\label{alg:check}
	\begin{algorithmic}[1]	
		\Procedure{WeightCheck}{$\beta^\star$}
		\State Fix a large $\beta^\star$.
		\State Use Algorithm \ref{alg:sample} to generate a sample $\mathbf{c}$.
		\If{$w_H(\mathbf{c}) = \omega$} output \texttt{Yes}
		\Else \ output \texttt{No}
		\EndIf
		\EndProcedure	
	\end{algorithmic}
\end{algorithm} 
\section{Numerical Examples}
\label{sec:numerics}
In this section, we shall apply a variant of Algorithm \ref{alg:estimate} to compute estimates of the weight enumerators, and Algorithm \ref{alg:check} to obtain estimates of the weight spectrum of specific moderate-blocklength RM codes. The method we use for computing weight enumerator estimates in this section (shown as Algorithm \ref{alg:estimateprac})
differs from Algorithm \ref{alg:estimate} in that we do not pick a value of $\ell$ (which determines the cooling schedule completely) in advance. Instead, we shall iterate the loop in Step 4 of Algorithm \ref{alg:estimate} and keep updating the estimate $\widehat{Z}_{\beta^\star}$ until it settles to within a precribed precision $\delta\in (0,1)$. Computer code for these algorithms, written in Julia, Python, and MATLAB, can be found at \cite{sjarnkgit}.

\begin{algorithm}[t]
	\caption{Estimating $A(\omega)$ via $\widehat{Z}$}
	\label{alg:estimateprac}
	\begin{algorithmic}[1]	
		\Procedure{Estimator}{}
		\State Fix a large $t\in \mathbb{N}$.
		\State Fix a (small) precision $\delta\in (0,1)$ and set $\beta \leftarrow 0$.
		\State Set \texttt{curr}$\ \leftarrow |\mathcal{C}|$ and \texttt{prev}$\ \leftarrow 0$.
		\While{$|\text{\texttt{curr}}-\text{\texttt{prev}}|>\delta$}
		\State Increment $\beta \leftarrow \beta+1/n$.
		\State Draw $t$ samples $\mathbf{c}_{1},\ldots,\mathbf{c}_{t}$ i.i.d. from $p_\beta$ using Algorithm \ref{alg:sample}.
		\State For $1\leq j\leq t$, set $X_{j} \leftarrow \text{exp}(-E(\mathbf{c}_{j})/n)$.
		\State Compute $Y = \frac{1}{t}\sum_{j=1}^t X_{j}$.
		\State Update \texttt{prev}$\ \leftarrow$ \texttt{curr} and \texttt{curr}$\ \leftarrow Y\cdot \text{\texttt{curr}}$.
		\EndWhile
		\State Output $\widehat{Z} = \text{\texttt{curr}}$.
		\EndProcedure	
	\end{algorithmic}
\end{algorithm} 

\subsection{Weight Enumerator Estimates}
Let us denote by $(A_{m,r}(\omega): 0\leq \omega\leq 2^m)$ the collection of weight enumerators, or equivalently, the weight distribution of RM$(m,r)$. We use the following fact (see, e.g., the survey \cite{rm_survey}) to ease computation: the weight distribution is symmetric about $\omega = 2^{m-1}$, i.e., $A_{m,r}(\omega) = A_{m,r}(n-\omega)$, for $0\leq \omega\leq 2^m$. {Using the fact that the dual code of RM$(m,\rho)$ is RM$(m,m-\rho-1)$, for $\rho\leq m-1$, we can obtain estimates of the weight distribution of RM$(m,r)$, for $ \left \lfloor \frac{m-1}{2}\right \rfloor < r \leq m-1$ by plugging in the estimates of the weight distribution of the corresponding dual code into MacWilliams' identities \cite{macwilliams}. We therefore confine our attention to computing estimates of the weight distributions of self-dual RM codes of the form RM$(m,\left \lfloor \frac{m-1}{2}\right \rfloor)$, for odd $m$, since these codes have the largest dimension among those codes RM$(m,r)$ where $r\leq \left \lfloor \frac{m-1}{2}\right \rfloor$.} {For such codes, it is known that those weight enumerators at weights not divisible by 4 are zero (see \cite[Cor. 13, Chap. 15, p. 447]{mws}). Hence, in the sequel, we shall only compute estimates of the weight enumerators of self-dual RM codes, at weights $2^{m-r}\leq \omega\leq 2^{m-1}$, {such that $\omega$ is divisible by $4$} (recall that the minimum distance of RM$(m,r)$ is $2^{m-r}$). We remark that for selected RM codes that are not self-dual, the knowledge of the weight spectrum (see, e.g., \cite{weight_spectrum}) can be used to inform the computation of weight estimates.} Furthermore, when we compare our size or rate estimates with the true weight enumerators $A_{m,r}(\omega)$ or rates $\frac{1}{2^m}\cdot {\log_2 A_{m,r}(\omega)}$, we shall confine our attention to only those weights with positive true weight enumerators. 

{Figure \ref{fig:rmwt_n9r4} shows comparisons of the rates of our estimates of the weight enumerators of RM$(9,4)$ with the true rates.} {The true weight enumerators for RM$(9,4)$ are taken directly from \cite{RM94_weights}.} We observe that our estimates are close to the ground truth. 

{We also use our method to estimate the rates of the weight enumerators of RM$(11,5)$, for weights $512\leq\omega\leq 1024$ that are multiples of 4; some of these estimates are tabulated in Table \ref{tab:rmwt115_1} and the remaining in Appendix \ref{sec:app-wtspec}. The computations were carried out on a computer with an Intel i7-7700 core and 16 GB of RAM.} We mention that computing the weight enumerator estimates for $\omega<512$ requires $\tau\geq 10^9$ in order to reach a codeword of the weight $\omega$; one may try to obtain estimates of $A(\omega)$ for such $\omega$, using more powerful computers.

\begin{table}[t!]
	\hfill
	\centering
		\begin{tabular}{||c||c||}
			\hline
			$\omega$ & $\frac{\log_2 \widehat{Z}}{2^m}$ \\ \hline 
			512 & 0.2967884396 \\ \hline
			516 & 0.3044142654 \\ \hline
			520 & 0.3098708781 \\ \hline
			524 & 0.3117907964 \\ \hline
			528 & 0.3159142454 \\ \hline
			532 & 0.3189211625 \\ \hline
			536 & 0.3210634545 \\ \hline
			540 & 0.3207781983 \\ \hline
			544 & 0.3292414942 \\ \hline
			548 & 0.3328099856 \\ \hline
			552 & 0.3325168244 \\ \hline
			556 & 0.33856337 \\ \hline
			560 & 0.3386807389 \\ \hline
			564 & 0.3458428641 \\ \hline
			568 & 0.3445511304 \\ \hline
			572 & 0.3516682886 \\ \hline
			576 & 0.3508534035 \\ \hline
			580 & 0.3532623147 \\ \hline
			584 & 0.3584547222 \\ \hline
			588 & 0.3586827759 \\ \hline
			592 & 0.3643411714 \\ \hline
			596 & 0.3626564667 \\ \hline
			600 & 0.3667665428 \\ \hline
			604 & 0.3719335599 \\ \hline
			608 & 0.369610572 \\ \hline
			612 & 0.3721375515 \\ \hline
			616 & 0.3791433248 \\ \hline
			620 & 0.378965573 \\ \hline
			624 & 0.3822856786 \\ \hline
			628 & 0.3860521584 \\ \hline
			632 & 0.3867559043 \\ \hline
			636 & 0.3904530961 \\ \hline
			640 & 0.3904831381 \\ \hline
			644 & 0.3886726124 \\ \hline
			648 & 0.3947849511 \\ \hline
			652 & 0.3994320755 \\ \hline
			656 & 0.4002439455 \\ \hline
			660 & 0.4041690765 \\ \hline
			664 & 0.405357908 \\ \hline
			
		\end{tabular}
\hfill
		\begin{tabular}{||c||c||}
			\hline
			$\omega$ & $\frac{\log_2 \widehat{Z}}{2^m}$ \\ \hline 
			668 & 0.4046384939 \\ \hline
			672 & 0.4083316113 \\ \hline
			676 & 0.4066164489 \\ \hline
			680 & 0.4122257465 \\ \hline
			684 & 0.416757749 \\ \hline
			688 & 0.4154593686 \\ \hline
			692 & 0.4151373483 \\ \hline
			696 & 0.4193250188 \\ \hline
			700 & 0.4228290541 \\ \hline
			704 & 0.426700815 \\ \hline
			708 & 0.4232395375 \\ \hline
			712 & 0.4260621686 \\ \hline
			716 & 0.4285488693 \\ \hline
			720 & 0.4313321598 \\ \hline
			724 & 0.4340572763 \\ \hline
			728 & 0.4378637509 \\ \hline
			732 & 0.4374155272 \\ \hline
			736 & 0.439222757 \\ \hline
			740 & 0.4383031672 \\ \hline
			744 & 0.4420132493 \\ \hline
			748 & 0.4419631323 \\ \hline
			752 & 0.4488207103 \\ \hline
			756 & 0.44347076 \\ \hline
			760 & 0.4476962448 \\ \hline
			764 & 0.4495450293 \\ \hline
			768 & 0.4489722934 \\ \hline
			772 & 0.4509841781 \\ \hline
			776 & 0.454131732 \\ \hline
			780 & 0.4562743742 \\ \hline
			784 & 0.4541577203 \\ \hline
			788 & 0.45629281 \\ \hline
			792 & 0.4579614556 \\ \hline
			796 & 0.4595784776 \\ \hline
			800 & 0.4604385389 \\ \hline
			804 & 0.4628383539 \\ \hline
			808 & 0.4655936825 \\ \hline
			812 & 0.4666761093 \\ \hline
			816 & 0.4660395012 \\ \hline
			820 & 0.4687878527 \\ \hline
			
		\end{tabular}
	\hfill \hfill
	\caption{Table of rate estimates $\frac{\log_2 \widehat{Z}}{2^m}$ of the weight enumerators for RM$(11,5)$, at weights $512\leq\omega\leq 820$ with non-zero weight enumerators. Here, the parameters $\tau = 10^6$, $t=10$, and $\delta = 0.001$.}
	\label{tab:rmwt115_1}
\end{table}

\begin{figure}[!h]
	\centering
	\includegraphics[width=\linewidth]{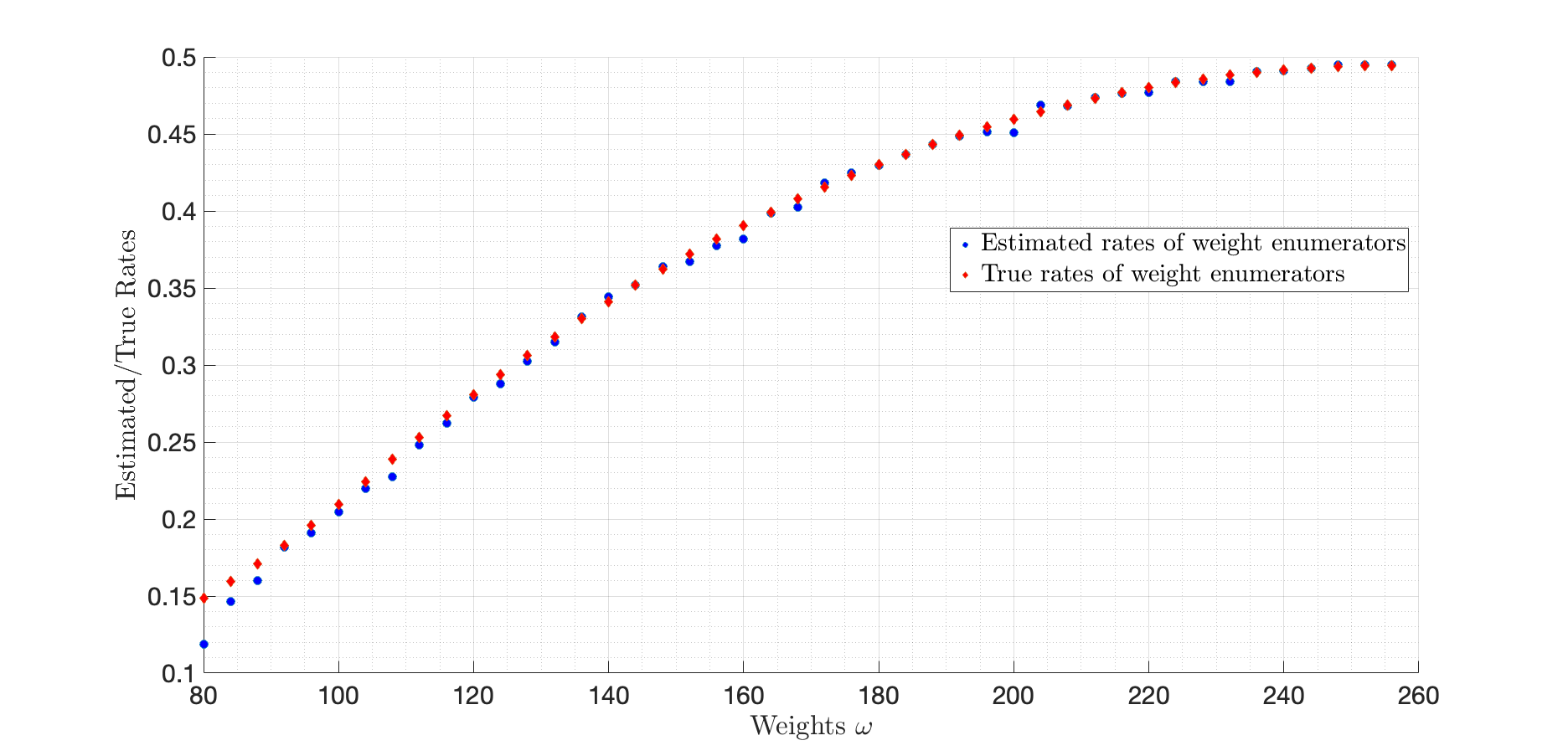}
	\caption{Plot comparing the estimates of rates of the weight enumerators (for selected weights with positive weight enumerators) obtained via our sampling-based approach with the true rates, for RM$(9,4)$, obtained from \cite{RM94_weights}. For these runs, we set $\tau = 5\times10^5$, $t=10$, and $\delta = 0.001$.}
	\label{fig:rmwt_n9r4}
\end{figure}


\subsection{Exact Weight Spectrum Computations}
Let us denote by $S_{m,r}$ the weight spectrum of RM$(m,r)$. From \cite[Cor. 13, Chap. 15, p. 447]{mws}, we have that $S_{m,r}\subseteq \{\omega: \omega \text{ is a multiple of } 2^{\left \lceil m/r\right \rceil - 1}\}$. {Furthermore, the exact weight enumerators are known for all weights $\omega < 2.5\cdot 2^{m-r}$; let $S_{m,r}^{<}$ denote the set of such weights with positive weight enumerators. It thus suffices to run Algorithm \ref{alg:check} for only those weights in the candidate set $\tilde{S}^{>}_{m,r}:=\{\omega: 2.5\cdot 2^{m-r}\leq\omega\leq 2^{m-1} \text{ and } \omega \text{ is a multiple of } 2^{\left \lceil m/r\right \rceil - 1}\}$, by the symmetry of the weight distribution.} We ran our algorithm for RM$(10,3)$ and RM$(10,4)$, for which the true weight spectra are still unknown. For each weight $\omega \in \tilde{S}_{m,r}$, we record the codeword, if any, found at that weight by Algorithm \ref{alg:check}. The message vectors $\mathbf{u}$ of length equal to the dimension of the code, which generate codewords of selected weights are presented in Appendix \ref{sec:codewords}; the complete list of codewords found at weights in $\tilde{S}_{m,r}$ can be found at \cite{sjarnkgit}. We observe that all weights in $\tilde{S}^>_{m,r}$ have positive weight enumerators, for $m=10$ and $r=3,4$, thereby giving rise to the following theorem.
\begin{theorem}
	We have that $S_{10,3} = S^<_{10,3}\cup \tilde{S}^>_{10,3}$ and $S_{10,4} =S^<_{10,4}\cup \tilde{S}^>_{10,4}$.
\end{theorem}
\section{Conclusion}
In this paper, we proposed a novel sampling-based approach for computing estimates of the weight enumerators and the weight spectra of the Reed-Muller (RM) family of codes. We observed that our estimates are close to the true sizes (or rates) for those RM codes where a direct computation of the true values is computationally feasible. Moreover, using our techniques, we obtained estimates of the weight enumerators of RM$(11,5)$, whose true weight enumerators are not known for all weights, and the exact, hitherto unknown weight spectra of RM$(10,3)$ and RM$(10,4)$. We also provided theoretical guarantees of the robustness of our estimates and argued that for a fixed error in approximation, our proposed algorithm uses a number of samples that is only polynomial in the blocklength of the code. Such sampling-based approaches have been largely unexplored for counting problems of interest in coding theory and we believe that there is much scope for the application of such techniques to other open problems.


%

\appendices
\section{Theoretical Guarantees}

\label{sec:theory}
To demonstrate how good our estimate $\widehat{Z}_{\beta^\star}$ of the weight enumerator $A(\omega)$ is, for a fixed $0\leq \omega\leq n$, we invoke the following (well-known) theorem of Dyer and Frieze \cite{dyerfrieze} (see also Theorem 2.1 in \cite{stefankovic}):
\begin{theorem}
	\label{thm:dyer}
	Fix an $\epsilon \geq 0$. Let $U_1,\ldots,U_\ell$ be independent random variables with $\mathbb{E}[U_i^2]/(\mathbb{E}[U_i])^2 \leq B$, for some $B\geq 0$ and for $1\leq i\leq \ell$. Set $\widehat{U} = \prod_{i=1}^\ell U_i$. Also, for $1\leq i\leq \ell$, let $V_i$ be the average of $16B\ell/\epsilon^2$ independent random samples having the same distribution as $U_i$; set $\widehat{V} =  \prod_{i=1}^\ell V_i$. Then,
	\[
	\Pr\left[(1-\epsilon)\mathbb{E}[\widehat{U}]\leq \widehat{V}\leq (1+\epsilon)\mathbb{E}[\widehat{U}]\right]\geq \frac34.
	\]
\end{theorem}
As a direct corollary, we obtain the following guarantee about our estimate $\widehat{Z}_{\beta^\star}$. Let $t^\star = 16e^2\ell/\epsilon^2$.
\begin{corollary}
	\label{cor:dyer}
	Fix an $\epsilon \geq 0$. If $Y_i$ is the average of $t^\star$ i.i.d. samples having the same distribution as $X_i$ (see (7) in the paper), then,
	\[
	\Pr\left[(1-\epsilon)Z_{\beta^\star}\leq \widehat{Z}_{\beta^\star}\leq (1+\epsilon)Z_{\beta^\star}\right]\geq \frac34.
	\]
\end{corollary}

\begin{proof}
	Since we have that $0\leq E(\mathbf{c})\leq n$ for all $\omega, n$, it follows that the random variable $X_i \in [e^{-1},1]$ as $\beta_i-\beta_{i-1} = \frac{1}{n}$. Hence, $\mathbb{E}[X_i^2]/(\mathbb{E}[X_i])^2 \leq e^2=:B$. The proof then follows from a simple application of Theorem \ref{thm:dyer} with the observation that by the independence of the $Y_i$s, we have $\mathbb{E}[Z_{\beta_0}\times \prod_{i=1}^\ell Y_i] = Z_{\beta^\star}$, from Eq. (7) in the paper.
\end{proof}
\begin{remark}
	Following Proposition 4.2 in \cite[Lecture 4]{sinclair}, we have that the constant on the right-hand side of Corollary \ref{cor:dyer} can be improved to $1-\gamma$ for $\gamma$ arbitrarily small, by using the new estimate $\overline{Z}_{\beta^\star}$ that is the median of $\widehat{Z}_{\beta^\star}^{(1)},\ldots,\widehat{Z}_{\beta^\star}^{(T)}$ where $T = O(\log \gamma^{-1})$ and each $\widehat{Z}_{\beta^\star}^{(i)}$, for $1\leq i\leq T$, is drawn i.i.d. according to (8) in the paper. Hence, we obtain an estimate that lies in $[(1-\epsilon)Z_{\beta^\star},(1+\epsilon)Z_{\beta^\star}]$, for $\epsilon$ arbitrarily small, with arbitrarily high probability.
\end{remark}
Observe that the number of samples $t^\star$ required to compute a single sample average as in (7) in the paper is polynomial (in fact, linear) in the length $\ell$ of the cooling schedule, for a fixed $\epsilon>0$. It thus remains to specify this length $\ell$. From arguments similar to that in \cite[Lecture 4]{sinclair} (see the paragraph following Eq. (4.6) there), we have that for $\beta^\star = O(n^2)$, the value $Z_{\beta^\star}$ is such that $Z_{\beta^\star} = (1+\delta_n)Z = (1+\delta_n)A(\omega)$, for $\delta_n = \text{exp}(-\Theta(n^2))$ and for any fixed weight $\omega$. In other words, for large $n$, with $\beta^\star = \Theta(n^2)$, our estimate $\widehat{Z}_{\beta^\star}$ is such that
\[
\Pr[(1-\epsilon)(1+\delta_n)A(\omega)\leq \widehat{Z}_{\beta^\star}\leq (1+\epsilon)(1+\delta_n)A(\omega)]\geq \frac34,
\]
from Corollary \ref{cor:dyer}. Hence, it suffices for $\ell$ to be $\Theta(n^3)$ (since $\beta_i - \beta_{i-1} = 1/n$, for $1\leq i\leq \ell$, with $\beta_0 = 0$ and $\beta_\ell = \beta^\star$) to obtain a good estimate of the true weight enumerator $A(\omega)$. We also then have that the total number of samples required, $t^\star\ell$, is $\Theta(n^6)$, for a fixed $\epsilon$, which is still only polynomial in the blocklength $n$ of $\mathcal{C}$, independent of its dimension. 

This must be contrasted with the number of computations required for brute-force search, which is at least $\min(2^k,2^{n-k})$ \cite{macwilliams}, that is exponential in $n$ when $k$ grows linearly in $n$. Furthermore, for the special case when $\mathcal{C}$ is a self-dual RM code of the form RM$(m,\frac{m-1}{2})$, for some $m\geq 1$ odd, we see that the time complexity of the algorithm in \cite{sarwate} is at least as much as the number of cosets of RM$(m-1,\frac{m-3}{2})$ in RM$(m-1,\frac{m-1}{2})$ (assuming that the weight enumerators of these cosets were all known in advance). This, in turn, equals $2^{\binom{m-1}{\frac{m-1}{2}}} = \text{exp}\left({\Theta\left(\frac{2^{m}}{\sqrt{m-1}}\right)}\right)$ (see Section 5.4 in \cite{asymp}), which is almost exponential in the blocklength $n=2^m$. Our sampling-based approach hence provides great savings in complexity, at the cost of some error in accuracy.
\section{Tables of Rate Estimates of Weight Enumerators of RM$(11,5)$}
\label{sec:app-wtspec}
We record our estimates of selected weight enumerators of RM$(11,5)$ in Table \ref{tab:rmwt115_2}.

\begin{table}[h!]
	\centering
	\hfill
	\begin{tabular}{||c||c||}
		\hline
		$\omega$ & $\frac{\log_2 \widehat{Z}}{2^m}$ \\ \hline 
		824 & 0.4677889949 \\ \hline
		828 & 0.4692737546 \\ \hline
		832 & 0.470439285 \\ \hline
		836 & 0.4720260806 \\ \hline
		840 & 0.4757967321 \\ \hline
		844 & 0.4728329498 \\ \hline
		848 & 0.4764918242 \\ \hline
		852 & 0.4758886476 \\ \hline
		856 & 0.4780183181 \\ \hline
		860 & 0.4790559436 \\ \hline
		864 & 0.4783843133 \\ \hline
		868 & 0.4801069772 \\ \hline
		872 & 0.4827332212 \\ \hline
		876 & 0.4821388907 \\ \hline
		880 & 0.4821725246 \\ \hline
		884 & 0.4830228041 \\ \hline
		888 & 0.4833586871 \\ \hline
		892 & 0.4838592857 \\ \hline
		896 & 0.4842506892 \\ \hline
		900 & 0.4869631419 \\ \hline
		904 & 0.4878321476 \\ \hline
		908 & 0.4893220263 \\ \hline
		912 & 0.4883336432 \\ \hline
		916 & 0.4889590025 \\ \hline
		920 & 0.4889086859 \\ \hline
		924 & 0.4898091244 \\ \hline

	\end{tabular}
	\hfill
	\begin{tabular}{||c||c||}
		\hline
		$\omega$ & $\frac{\log_2 \widehat{Z}}{2^m}$ \\ \hline
		
		928 & 0.491859122 \\ \hline
		932 & 0.4903585543 \\ \hline
		936 & 0.4922892439 \\ \hline
		940 & 0.4914607835 \\ \hline
		944 & 0.4934381041 \\ \hline
		948 & 0.4931433172 \\ \hline
		952 & 0.4945423179 \\ \hline
		956 & 0.4929904748 \\ \hline
		960 & 0.4943895946 \\ \hline
		964 & 0.4952178539 \\ \hline
		968 & 0.4945798255 \\ \hline
		972 & 0.495527873 \\ \hline
		976 & 0.4970328654 \\ \hline
		980 & 0.4966750273 \\ \hline
		984 & 0.4954571397 \\ \hline
		988 & 0.4963111889 \\ \hline
		992 & 0.4968925105 \\ \hline
		996 & 0.4972624121 \\ \hline
		1000 & 0.496628302 \\ \hline
		1004 & 0.4974914362 \\ \hline
		1008 & 0.4975938289 \\ \hline
		1012 & 0.4975072018 \\ \hline
		1016 & 0.4966127403 \\ \hline
		1020 & 0.4978946567 \\ \hline
		1024 & 0.4980060621 \\ \hline
	\end{tabular}
	\hfill \hfill
	\caption{Table of rate estimates $\frac{\log_2 \widehat{Z}}{2^m}$ of the weight enumerators for RM$(11,5)$, at weights $821\leq\omega\leq 1024$ with non-zero weight enumerators. Here, the parameters $\tau = 10^6$, $t=10$, and $\delta = 0.001$.}
	\label{tab:rmwt115_2}
\end{table}

\section{Tables of Selected Message Vectors}
\label{sec:codewords}
\begin{table}[t!]
	\begin{tabular}{||p{0.07\textwidth} || p{0.35\textwidth}||}
		\hline
		Weight $\omega$ & Support of message vector \\ \hline
		328 & \{1, 9, 
		11, 14, 15, 18, 20, 23, 
		24, 25, 27, 31, 34, 35, 36, 43, 47, 48, 49, 50, 54, 59, 61
		, 62, 65, 67, 69, 
		71, 72, 73, 79, 82, 83, 84, 86, 87, 91, 94, 95, 98, 99, 100, 102, 107, 108, 110, 111, 112, 113, 116, 117, 118, 120, 121, 122, 123, 124, 126, 
		129, 130, 131, 132, 135, 136, 137, 138, 140, 141, 143, 144, 145, 146, 147, 148, 149, 150, 151, 153, 154, 155, 159, 160, 163, 164
		, 166, 172, 173, 176\}\\ \hline
		480 & \{7, 12, 13, 14, 15, 17, 19, 20, 22, 27, 28, 29, 30, 31, 32, 35, 36, 38, 39, 41, 42, 48, 50, 54, 55, 56
		, 57, 60, 67, 68, 70, 72, 73, 75, 76, 80, 82, 83, 85, 86, 88, 91, 93, 95, 96, 97, 98, 99, 100, 102, 103, 104, 107, 109, 111, 116, 118, 119, 121, 123, 124, 126, 128, 129, 130, 131, 133, 
		139, 142, 143, 144, 148, 150, 151, 156, 162, 164, 165, 166, 167, 168, 169, 170, 171, 173, 174, 175, 176\} \\ \hline
		512 & \{1, 3, 6, 8, 12, 14, 15, 16, 18, 19, 20, 25, 27, 30, 32, 35
		, 36, 37, 38, 40, 41, 42, 43, 44, 47, 48, 50, 51, 52, 54, 55, 56, 59, 62, 63, 66, 67, 71, 72, 77, 78, 79, 82, 85, 87, 88, 89, 91, 92, 96, 99, 100, 101, 102, 103, 104, 105, 106, 108, 110, 113, 114, 115, 116, 120, 121, 123, 125, 126, 127, 128, 129, 130, 132, 133,
		135, 136, 140, 141, 143, 147, 149, 150, 151, 152, 154, 155, 161, 164, 165, 167, 168, 169, 170, 171, 172, 173, 175, 176\} \\ \hline
	\end{tabular}
	\caption{Table of supports of message vectors corresponding to codewords of RM$(10,3)$ of selected weights}
	\label{tab:rmsp10_3}
\end{table}
In this section, we list the ``message vectors'' $\mathbf{u}$ of length equal to the dimension of the RM code under consideration, which give rise to the codewords $\mathbf{c} = \mathbf{u}G$ of selected weights obtained via a sampling-based search strategy; here $G$ is a chosen generator matrix of the RM code. We next specify the construction of the generator matrix $G = G_{m,r}$ of the RM$(m,r)$. For $m\geq 1$ and $0\leq r\leq m$, we set $G_{m,r}$ equal to the $2^m$-length all-ones vector $1^{2^m}$, if $r=0$, and $G_{m,m} = I_{2^m}$, where $I_n$ denotes the $n\times n$ identity matrix. For $1\leq r\leq m-1$, $G_{m,r}$ is constructed recursively as follows:
\[
G_{m,r} = 
\begin{pmatrix}
	G_{m-1,r} & G_{m-1,r}\\
	\text{\large\underline{0}} &  G_{m-1,r-1}
\end{pmatrix}.
\]
In the above construction, $\text{\large \underline{0}} $ represents the all-zero matrix of order $\binom{m-1}{\leq r-1}\times 2^{m-1}$.

\begin{table}[t!]
	\begin{tabular}{||p{0.07\textwidth} || p{0.35\textwidth}||}
		\hline
		Weight $\omega$ & Support of message vector \\ \hline
		
		164 & \{2, 3, 8, 19, 22, 24, 27, 29, 30, 32, 33, 34, 39, 40, 41, 42
		, 44, 45, 47, 49, 51, 52, 53, 54, 59, 61, 65, 69, 76, 77, 78, 79, 81, 84, 85, 88, 90, 91, 93, 94, 95, 99, 101, 104, 107, 109, 112, 113, 117, 119, 120, 123, 124, 126, 127, 131, 133, 134, 135, 136,
		137, 141, 143, 144, 145, 146, 152, 155, 159, 160,
		161, 165, 171, 172, 176, 178, 184, 186, 187, 188, 190, 191, 195, 197, 199, 201, 205, 207, 209, 211, 213, 214, 216, 217, 219, 222, 223, 224, 226, 230, 232, 
		233, 235, 236, 240, 241, 242, 243, 244, 247, 248, 252, 253, 254, 255, 259, 261, 264, 267, 269, 273, 275, 276, 279, 281, 282, 283, 285, 286, 288, 289, 292, 293, 294, 295, 296, 302, 303, 305, 307, 308, 309, 311, 312, 314, 315, 317, 
		319, 327, 333, 336, 338, 341, 345, 346, 347, 349, 350, 351, 355, 361, 362, 365, 367, 372, 374, 379, 382, 383, 384, 385, 386\}\\ \hline
		216 & \{7, 9, 10, 19, 23, 25, 26, 27, 29, 31, 32, 38, 42, 47, 49, 57, 58, 59, 64, 69, 72, 73, 74,
		77, 82, 84, 85, 88, 92, 93, 96, 98, 99, 106, 109, 110, 113, 117, 127, 130, 134, 135, 136, 138, 139, 143, 147, 148, 150, 154, 155, 158, 162, 165, 166, 172, 174, 176, 179, 182, 183, 185, 187, 190, 192,
		194, 202, 204, 205, 208, 209, 212, 213, 221, 223, 226, 227, 229, 231, 233, 237, 238, 239, 240, 242, 243, 245, 248, 249, 258, 259, 263, 265
		, 267, 269, 271, 275, 276, 278, 280, 282, 283, 285, 287, 289, 294, 295, 296, 297, 298, 299, 303, 307, 312, 313, 315, 316, 317, 318, 319, 320, 322, 323, 326, 327, 333, 335, 336, 337, 338, 341, 342, 348, 349, 351, 352, 355, 356, 358, 359, 360, 361, 367, 368, 369, 371, 373, 377, 378, 380, 382, 383, 386\} \\ \hline
		512 & \{4, 5, 7, 8, 9, 16, 22, 23, 25, 28, 29, 31, 32, 36, 37, 38, 39, 40, 41, 42, 43, 48, 50, 51, 52, 54, 55, 56, 57, 59, 61, 64, 66, 67, 68, 70, 72, 73, 74, 81, 82,
		83, 90, 93, 94, 103, 104, 105, 107, 109, 110, 112, 113, 114, 118, 121, 124, 125, 128, 130, 131, 133, 134, 135, 136, 137, 138,
		141, 143, 144, 145, 146, 152, 153, 155, 157, 161, 163, 165, 166, 168, 169, 170, 172, 175, 177, 178, 179, 180, 182, 183, 184, 186, 187,
		188, 189, 191, 193, 200, 201, 202, 203, 204, 206, 207, 208, 209, 211, 213, 214, 216, 217, 218, 219, 224, 226, 230, 231, 235, 237, 243, 245, 247, 248, 254, 255, 256, 259, 264, 266, 267, 272, 273, 282, 
		283, 284, 285, 286, 290, 292, 293, 295, 297, 298, 300, 303, 306, 308, 312, 313, 316, 317, 318, 319, 320, 321, 323, 324, 325, 329, 
		331, 333, 334, 337, 341, 343, 346, 347, 349, 350, 351, 352, 353, 355, 357, 358, 359, 360, 364, 367, 368, 370, 371, 372, 373, 374, 375, 376, 378, 381, 382, 383, 384, 385, 386\} \\ \hline
	\end{tabular}
	\caption{Table of supports of message vectors corresponding to codewords of RM$(10,4)$ of selected weights}
	\label{tab:rmsp10_4}
\end{table}

Now, given the matrix $G_{m,r}$ as above, we index its columns using length-$m$ binary vectors of the form $(b_1,\ldots,b_m)\in \{0,1\}^m$ in the lexicographic order. In other words, we label the $i^\text{th}$ column, for $1\leq i\leq 2^m$, using the length-$m$ binary representation of $i-1$. Consider then the following set of columns:
\[
\mathcal{I}_{m,r}:=\{\mathbf{b}=(b_1,\ldots,b_m)\in \{0,1\}^m:\ w_H(\mathbf{b})\leq r\}.
\]

From \cite[Lemma 20]{arnk23tit}, we see that the (square) matrix $\overline{G}_{m,r}$ that consists of columns of $G_{m,r}$ in $\mathcal{I}_{m,r}$ is full rank. Hence, given a codeword $\mathbf{c}$ of a selected weight, its corresponding message vector $\mathbf{u}$ is obtained as
\[
\mathbf{u} = \mathbf{c}(\mathcal{I}_{m,r})\cdot(\overline{G}_{m,r})^{-1}.
\]
In the above equation, $\mathbf{c}(\mathcal{I}_{m,r})$ denotes those symbols in $\mathbf{c}$ at locations indexed by the vectors in $\mathcal{I}_{m,r}$.

In Tables \ref{tab:rmsp10_3} and \ref{tab:rmsp10_4}, due to space constraints, we specify the supports of $\mathbf{u}\in \{0,1\}^{k}$, where $k = \dim (\text{RM}(m,r))$, corresponding to codewords of RM$(10,3)$ and RM$(10,4)$, respectively, of selected weights.


\ifCLASSOPTIONcaptionsoff
  \newpage
\fi



%
%
\bibliographystyle{IEEEtran}
{\footnotesize
	\bibliography{references.bib}}

%




\end{document}